\documentclass[a4paper,11pt,oneside]{article}
\usepackage[english]{babel}
\usepackage[utf8]{inputenc}
\usepackage[T1]{fontenc}
\usepackage{amssymb,amsmath,amsthm,bm,bbold} 
\usepackage{enumerate}
\usepackage{latexsym}
\usepackage{mathtools}
\usepackage{mathrsfs}
\usepackage{fancyhdr}
\usepackage{setspace}
\usepackage{color} 
\usepackage{graphicx} 
\usepackage{subfigure}
\usepackage{caption}
\usepackage{dirtytalk}
\usepackage{stmaryrd}
\usepackage{geometry}
\usepackage[pdfpagelabels,plainpages=false]{hyperref}
\hypersetup{colorlinks=false}
\usepackage{bm}
\usepackage{upgreek}

\newcommand{\R}{\mathbb R}

\newcommand{\n}{\noindent}

\newcommand{\f}{\frac}
\newcommand{\vs}{\vspace{0.5cm}}

\newcommand{\ba}{\begin{eqnarray}} 
\newcommand{\ea}{\end{eqnarray}}
\newcommand{\bdm}{\begin{displaymath}}
\newcommand{\edm}{\end{displaymath}}
\newcommand{\bdn}{\begin{eqnarray}}
\newcommand{\edn}{\end{eqnarray}}
\newcommand{\bay}{\begin{array}{c}}
\newcommand{\eay}{\end{array}}
\newcommand{\ben}{\begin{enumerate}}
\newcommand{\een}{\end{enumerate}}
\newcommand{\beq}{\begin{equation}}
\newcommand{\eeq}{\end{equation}}

\newcommand{\be}{\begin{equation}}
\newcommand{\ee}{\end{equation}}

\newcommand{\mg}{\mathcal{G}^{\lambda}} 
\newcommand{\mgm}{\mathcal{G}^{\mu}}

\newcommand{\x}{{\bf x}}
\newcommand{\y}{{\bf y}}

\newcommand{\p}{{\bf p}}
\renewcommand{\k}{{\bf k}}

\newcommand{\s}{{\bf s}}

\newtheorem{pro}{Proposition}[section]

\newcommand{\ve}{\varepsilon}

\newcommand{\bC}{{\mathbb C}}

\input epsf

\newcommand{\donothing}[1]{}

\usepackage{url}

\title{On the Hamiltonian for three bosons \\with point interactions}

\author{Rodolfo Figari$^1$, Alessandro Teta$^2$ \\
\\
$^1$INFN, Sezione di Napoli, Complesso Universitario di Monte S. Angelo,\\
Via Cintia Edificio 6, 80126 Napoli, Italy\\
\\
$^2$Dipartimento di Matematica G. Castelnuovo, Sapienza Universit\`a di Roma,\\ 
Piazzale  Aldo Moro 5,  00185 Roma, Italy}

\makeatletter
\let\orgdescriptionlabel\descriptionlabel
\renewcommand*{\descriptionlabel}[1]{%
  \let\orglabel\label
  \let\label\@gobble
  \phantomsection
  \edef\@currentlabel{#1}%
  \let\label\orglabel
  \orgdescriptionlabel{#1}%
}
\makeatother

\begin{document}
\maketitle

\hfill {\em Dedicated to Sergio Albeverio}

\vs

\begin{abstract} We briefly summarize the most relevant steps in the search of rigorous results  about the properties of  quantum systems made of three bosons interacting with zero-range forces. We also describe recent attempts to solve the unboundedness problem of point-interaction Hamiltonians for a three-boson system, keeping unaltered the spectrum structure at low energies.
\end{abstract}

\section{Introduction}

In the unbounded scientific production of Sergio Albeverio the unboundedness of the zero-range Hamiltonians for three bosons and the Efimov effect play a very special role. First of all, as he pointed out in \cite{alfi}, because of the connection of the unboundedness problem with the existence of a non-trivial self-interacting relativistic quantum field theory, which was one of the main interest of his early scientific career. 
On the other hand, the peculiar structure of the spectrum at low negative energies of such Hamiltonians  (suggesting the existence of the so called Efimov trimers) was for him a challenge to analize rigorously the peculiar discrete scaling of the eigenvalues of zero-range multi-particle Hamiltonians.
The physics and mathematical-physics literature on the three-boson quantum system and the Efimov effect is nowadays so extensive that we cannot claim that the reader will find in this contribution a thorough summary of the subject and a comprehensive list of references. We will try first to describe the different attempts made to solve the unboundedness problem in a rigorous way and we will mention some clever suggestions about the solution of the problem which were never followed up with final results. On the basis of those suggestions, we tried recently to define zero-range Hamiltonians for the system of three bosons either breaking the rotational symmetry or adding three-particle contributions. We outline here methodology and results of this work-in-progress, examining in  particular the effect of the three-particle interaction, suggested in the past by Minlos, Faddeev, Albeverio, Hoegh-Krohn and Wu,  on the spectral structure of the Hamiltonians.\\
We want first to give an outline of how the quantum three-body problem appeared in the physical literature.

\section{The Thomas paper}

\noindent
In 1935, on G.E. Uhlenbeck's suggestion, L.H. Thomas investigated the interaction between a neutron and a proton in order to analize the structure of Tritium nucleus  \cite{TH}. He made the assumption of negligible interaction between the two neutrons and examined how short the neutron-proton interaction could be. He proved that the energy of the system was not bounded below for shorter and shorter interaction range. His conclusions are clearly stated in the abstract of the paper:  ``... .We conclude that: either two neutrons repel one another by an amount not negligible compared with the attraction between a neutron and a proton; or that the wave function cannot be symmetrical in their positions; or else that the interaction between a neutron and a proton is not confined within a relative distance very small compared with $10^{-13} cm$ '' .\\
Notice that few years later  the connection between spin and statistics was finally clarified. It then became clear that the wave function of the two neutrons could not be symmetrical under the exchange of their positions. Nevertheless, Thomas result  indicated that zero-range interactions, while perfectly defined for a two-particle system, allowing a finite energy for the ground state, seem to give Hamiltonians unbounded from below in the three-boson case. Since then, this effect of ``falling to the center'' of the quantum three-body system with zero-range interactions is known as Thomas effect (or Thomas collapse).\\
\noindent
In the paper, Thomas studied, in the center of mass reference frame, the eigenvalue problem for the free Hamiltonian outside a small region around the origin, where the three particles occupy the same position. He found that for negative energy values it was possible to exhibit  a square integrable solution showing singularities on the planes where the two neutron positions coincide. The solution was successively used as a test function to show that the energy of the ground state of the system was unbounded from below if the range of the interaction (a sum of two-body potentials or a singular boundary condition) was made shorter and shorter.\\
In particular, let $\x_{1}, \x_{2}$ the position of the neutrons and $\x_{3}$ the position of the proton. In the center of mass reference frame let $\s_{i} = \x_{i}  -\x_{3}$ for $i =1,2$ be the relative coordinates of the neutrons with respect to the proton. The free Hamiltonian eigenvalue equation then reads

\begin{equation}
- \left(  \,\f{\hbar}{4 \pi^{2}m} \right) (\Delta_{\s_{1}} +  \Delta_{\s_{2}} + 
\nabla_{\s_{1}} \cdot 
\nabla_{\s_{2}}) \Psi (\s_{1}, \s_{2}) + \mu^{2} \Psi (\s_{1}, \s_{2}) = 0\,\,\,\,\,\,\,\,\,\,\,\,\,\,\, \mu > 0\,.
\label{thomas1}
\end{equation}

\noindent
Thomas found that a singular solution of (\ref{thomas1}) is given by

\begin {equation}
\Psi (\s_{1}, \s_{2})  = \frac{1}{s^{2}} K_{0} (\eta s) \left[\frac{\frac{\pi}{2} - \arctan\xi_{1}}{\xi_{1} (1 + \xi_{1}^{2})} + \frac {\frac{\pi}{2} - \arctan\xi_{2}}{\xi_{2} (1 + \xi_{2}^{2})}  \right]
\label{thomas2}
\end{equation}

\noindent
where
\begin{equation}
 s^{2} = |\s_{1}|^{2} + |\s_{2}|^{2} - \s_{1} \cdot \s_{2}; \,\,\,\,\,\,\,\,\, \xi_1 = \frac{|\s_1|}{|\s_1 - 2 \s_2|}; \,\,\,\,\,\,\,\,\,\,\,\,\,\,  \xi_2 = \frac{|\s_2|}{|\s_2 - 2 \s_1|},
\end{equation}
$K_{0} (x)$ is the zero-th Macdonald function of integer order (see e.g.\cite{er}) and $\displaystyle \eta = \left(\f{4 \pi^{2}m }{\hbar}\right)^{1/2}\mu$.\\

\noindent
 We list few important features of the solution.  

\noindent
The behaviour of the solution close to the coincidence plane $\{\s_{1} = 0\}$ when  $|\s_{2}| > 0 $ is 
\begin{equation}
\Psi (\s_{1}, \s_{2}) \approx  \frac{1}{|\s_{1}|} \left(\frac{\pi}{\sqrt{3}}\frac{1}{ |\s_{2}|} K_{0} (\eta |\s_2|) \right).
\end{equation}

\noindent
Symmetrically, when $|\s_{1}| \cong 0$ e $|\s_{2}| > 0 $ 
\begin{equation}
\Psi (\s_{1}, \s_{2}) \approx  \frac{1}{|\s_{2}|} \left(\frac{\pi}{\sqrt{3}}\frac{1}{ |\s_{1}|} K_{0} (\eta |\s_1|) \right).
\end{equation}

\noindent
Moreover, the following scale transformation send eigenvectors in eigenvectors with different $\mu$

\begin{equation}
\Psi_{\lambda} (\s_{1}, \s_{2}) \equiv \lambda^{3} \Psi (\lambda \s_{1}, \lambda \s_{2}) =  \frac{\lambda}{s^{2}} K_{0} (\lambda \eta s) \left[\frac{\frac{\pi}{2} - \arctan\xi_{1}}{\xi_{1} (1 + \xi_{1}^{2})} + \frac {\frac{\pi}{2} - \arctan\xi_{2}}{\xi_{2} (1 + \xi_{2}^{2})}  \right]
\end{equation}
with $\lambda > 0$ and $\|\Psi_{\lambda}\|_2    = \|\Psi\|_2$.

\noindent
Notice that the singularity shown by the solution close to the coincidence planes are the same of the potential of a density charge distributed on the planes. As we will see this behaviour is typical of the functions in the domain of the zero-range interaction Hamiltonians suggested by Skorniakov, Ter-Martirosian \cite{Tms} and Danilov \cite{dan}  many years after the work of L. H. Thomas.

\section{Zero-range interaction Hamiltonians}

\n
It is well known that for system of quantum particles in $\R^3$ at low temperature one has that the thermal wavelength 
\[
\lambda =\f{h}{p} \; \simeq \;  \hbar \,\sqrt{\f{2 \pi}{m k_B T}} 
\]
is much larger than the range of the two-body interaction. 
Under these conditions the system exhibits a universal behavior, 
i.e., relevant observables do  not depend on the details of the two-body interaction, but only on a single physical parameter known as 
 the scattering length 
 \[
 a := \lim_{|\k|\rightarrow 0} f(\k,\k')
\]
where $ f(\k,\k')$ is the scattering amplitude of the two-body scattering process. 
In this regime it is natural to expect that the qualitative behavior of the system of particles in $\R^3$ is well described by the (formal) effective Hamiltonian with zero-range interactions 
\be\label{forh}
  ``H = -\sum_{i} \f{1}{2m_i} \Delta_{\x_i} + \sum_{i<j} a_{ij} \delta (\x_i - \x_j) \, \text{''}
\ee
where $m_i$ is the mass of the $i$-th particle and $a_{ij}$ is the scattering length associated to the scattering process between the particles $i$ and $j$. Hamiltonians of the type (\ref{forh}) are widely used in the physical literature to investigate, e.g., the behaviour of ultra-cold gases (see, e.g., \cite{bh}, \cite{cmp}, \cite{ct}, \cite{cw}, \cite{km}, \cite{ne}, \cite{wc1}, \cite{wc2}).

\n
From the mathematical point of view the first step is to give a rigorous meaning to \eqref{forh}. We define the Hamiltonian for the system of $n$ particles with two-body zero-range interactions as a non trivial  self-adjoint (s.a.) extension in $L^2(\R^{3n})$ of the operator
\be\label{fre0}
 -\sum_{i} \f{1}{2m_i} \Delta_{\x_i}
\ee
 defined on $H^2$-functions vanishing on each hyperplane $\{\x_i =\x_j\}$.

\n
The complete characterization of  such Hamiltonians can be  obtained in the simple case $n=2$. Indeed, in the center of mass reference frame, denoting with  $\x$  the relative coordinate, one has to consider the operator in $L^2(\R^3)$
\be
\tilde{h}_0 = - \f{1}{2 m_{12}} \Delta_{\x} \,, \;\;\;\;\;\;\;\; m_{12}=\f{m_1 m_2}{m_1 + m_2} 
\ee
defined on $H^2$-functions vanishing at the origin. Such  operator has defect indices $(1,1)$ and the one-parameter family of s.a. extensions $h_{\alpha}$, $\alpha \in \R$,  can be explicitly constructed. Roughly speaking,  $h_{\alpha}$ acts as the free Hamiltonian outside the origin and  $\psi \in D(h_{\alpha})$  satisfies the (singular) boundary condition at the origin

\be\label{bc1d}
\psi(\x) = \f{q}{|\x|} \; +\; \alpha\,  q \;+\;  o(1) \;\;\;\;\; \text{for}\;\;\; |x| \rightarrow 0
\ee

\n
where $q$ is a constant depending on $\psi$ and $\alpha=a^{-1}$. Moreover, the resolvent  $(h_{\alpha}-z)^{-1}$, $z \in \bC \setminus \R$, can be explicitly computed and all the spectral properties of $h_{\alpha}$ characterized (see \cite{al}, \cite{al4}).

\n
For systems of $n$ particles, with $n>2$, the construction is much more difficult because of the presence  of infinite dimensional defect spaces (see, e.g., \cite{al2}, \cite{al3}, \cite{fcft1}, \cite{bft}, \cite{bt}, \cite{CDFMT}, \cite{CDFMT2}, \cite{cft}, \cite{DFT}, \cite{dim2}, \cite{FT}, \cite{gl1}, \cite{gl2}, \cite{mp}, \cite{ms}, \cite{MM}, \cite{m1}, \cite{m3}, \cite{m4}, \cite{mf1}, \cite{mf2}, \cite{ms1}, \cite{ms2}, \cite{th}).  To simplify notation, we describe the problem, and some previous attempts to solve it, in the case of  three identical bosons with masses $1/2$,  in the center of mass reference frame. 

\n
Let $\x_1, \x_2$ and $ \x_3 = -\x_1 - \x_2$ be the cartesian coordinates of the particles. 
Let us introduce the Jacobi coordinates   
\beq
\x  = \x_2 - \x_3\,, \;\;\;\;\;\; \y =\f{1}{2} (\x_2 + \x_3) - \x_1\,
\eeq

\n
with inverse given by $
\x_1= -\f{2}{3} \y\,, \;\;  \x_2=\f{1}{2} \x+ \f{1}{3} \y \,, \;\;\; \x_3=- \f{1}{2} \x+ \f{1}{3} \y\,. $ 
Due to the symmetry constraint, the Hilbert space of states is 
\beq\label{hspa}
L^2_s(\R^6) \!=\! \Big\{ \psi \in L^2(\R^6)\,  \;\text{s.t.}\;  \psi(\x,\y)= \psi(-\x,\y)= 
\psi\Big( \f{1}{2} \x + \y,  \f{3}{4} \x-\f{1}{2} \y   \Big) \Big\}
\eeq

\n
and the formal Hamiltonian reads

\beq\label{fham}
 ``H= - \Delta_{\x} - \f{3}{4}\Delta_{\y} +  a\delta(\x) +  a\delta (\y -  \x/2) +  a\delta (\y +  \x/2)\, \text{''},
\eeq

\n
i.e., $H$ is a perturbation of the free dynamics in $\R^6$ supported by the three-dimensional hyperplanes 
\be
\Sigma=\{\x=0\} \cup  \{\y-\x/2=0\} \cup \{\y+\x/2=0\}\,. 
\ee
According to our mathematical definition, a s.a. Hamiltonian in $L^2_s(\R^6)$  corresponding to the formal operator $H$ is  a non trivial s.a. extension of the operator
\begin{align}\label{h0p}
&\tilde{H}_0 =- \Delta_{\x} - \f{3}{4}\Delta_{\y}\,, \;\; D(\tilde{H}_0)=\Big\{  \psi \in L^2_s(\R^6) \; \text{s.t.}\; \psi \in H^2(\R^6) ,  \;\psi\big|_{\Sigma} =0 \Big\}\,.
\end{align}
\n
As we already mentioned the defect indices are now infinite  and the problem arises of how to choose and characterize a class of s.a. extensions with the right physical properties.

\n
An apparently reasonable choice based on the analogy with the case $n=2$ is due to Ter-Martirosian and Skorniakov \cite{Tms}. Indeed, they defined an operator $H_{\alpha}$ acting as  the free Hamiltonian outside the hyperplanes  and satisfying a boundary condition at the hyperplanes. Specifically, they impose

\beq\label{bc23}
\psi(\x,\y) = \f{\xi(\y)}{|\x|} + \alpha\, \xi(\y) + o(1) \,, \;\;\;\; \text{for}\;\; |\x|\rightarrow 0\, \;\; \text{and}\;\;\y \neq 0
\eeq

\n
where $\xi$ is a function depending on $\psi$. Due to the bosonic symmetry, the same conditions  for  $|\y-\x/2| \rightarrow 0$ and $|\y+\x/2| \rightarrow 0$ have to be satisfied.

\n
We note that the first term in the right-hand side of (\ref{bc23}) coincides with the first term of the asymptotic expansion of the potential produced by the charge density $\xi$ distributed on $\{ \x = 0 \}$. Therefore, an equivalent way to describe a wave function $\psi$ in the domain of $H_{\alpha}$ is the following. Any $\psi \in D(H_{\alpha})$ can be decomposed as 
\be
 \psi=w^{\lambda} + \mg \xi\,, \;\;\;\;\;\; w^{\lambda} \in H^2(\R^6)
 \ee
 where $\lambda>0$ and 
\beq
 \widehat{\mg \xi}  (\k,\p)  =
\sqrt{\f{2}{\pi}} \, \, \f{\hat{\xi}(\p)  + \hat{\xi}(\k - \f{1}{2} \p) + \hat{\xi}(-\k - \f{1}{2} \p  )    }{|\k|^2 + \f{3}{4} |\p|^2 +\lambda}\,. 
\eeq
Note that the function $\mg \xi(\x,\y)$ has the asymptotic behaviour
\begin{align}\label{abG}
\mathcal G^{\lambda} \xi  (\x,\y)=& \f{\xi(\y)}{|\x|} - \f{1}{(2\pi)^{3/2}} \int\!\!d\p\, e^{i \p \cdot \y} 
\big( T^{\lambda} \hat{\xi} \,\big) (\p)+o(1) 
\end{align}

\n
 for $|\x|\rightarrow 0$ and $\y \neq 0$, where
\beq\label{opga}
\big(T^{\lambda} \hat{\xi} \,\big) (\p) :=   \sqrt{\f{3}{4} |\p|^2 \!+\! \lambda}  \,\, \hat{\xi}(\p) - \f{1}{\pi^2} \!\!\int\!\! d\p'\, \f{ \hat{\xi}(\p')}{|\p|^2 + |\p'|^2 + \p\cdot \p' +\lambda} \,.
\eeq
We will refer to the first and second term in (\ref{opga}) respectively as the diagonal and the non-diagonal part of $T^{\lambda}$.\\
Therefore the boundary condition \eqref{bc23} is now rewritten as 
\be\label{bc22}
\alpha \,\hat{\xi}(\p)+\big( T \,\hat{\xi} \,\big) (\p) \;=\;  (w^{\lambda}\big|_{\x =0})^{\!\wedge} (\p)\,.
\ee
There is an ambiguity in the above definition since the domain of the symmetric and unbounded operator $T^{\lambda} $ in $L^2(\R^3)$ is not specified. As a first attempt one can choose 
\be
D(T^{\lambda} )= \{ \hat{\xi}\in L^2(\R^3) \;|\; \int\!\! d\k\, |\k|^2 |\hat{\xi} (\k)|^2 < \infty \} \equiv \{\hat{\xi} \;| \; \xi \in H^1(\R^3) \}\,.
\ee
Note that for $\hat{\xi} \in D(T^{\lambda} )$  both terms in the r.h.s.  of \eqref{opga} belong to $L^2(\R^3)$ (see, e.g., \cite{FT}). 

\n
As a matter of fact, the operator $H_{\alpha}$ defined in this way is symmetric but not  s.a. and it turns out that its s.a.  extensions are all unbounded from below. This fact, first noted by Danilov \cite{dan}, was rigorously analyzed by Minlos and Faddeev in  \cite{mf1}, \cite{mf2}.


\section{{\bf Minlos and Faddeev contributions}}

In the first  of their seminal contributions (see \cite{mf1}) Minlos and Faddeev consider a system of three bosons and approach the  general mathematical problem to give a meaning to the formal Hamiltonian with zero-range interactions.  Working in momentum space and using the  theory of s.a. extensions of semibounded operators developed by Birman \cite{bir}, they obtain the abstract   characterization of all the s.a. extensions of the operator \eqref{h0p}. 

\n
Then they observe  that the s.a. extensions of the operator $H_{\alpha}$ in $L^2_s(\R^6)$ introduced by Ter-Martirosian and Skornyakov are in one-to-one correspondence with the s.a. extensions of the operator $T^{\lambda}$ in $L^2(\R^3)$.  Such  operator $T^{\lambda} $  defined on $D(T^{\lambda} )$ has defect indices $(1,1)$ and they find that a  s.a. extension $T^{\lambda}_{\beta}$, $\beta \in \R$, of $T^{\lambda} $ is defined on 
\begin{align}\label{dtb}
D(T^{\lambda}_{\beta}) = \Big\{ \hat{\xi} \in L&^2(\R^3) \;|\; \hat{\xi}= \hat{\xi}_1 + \hat{\xi}_2, \; \hat{\xi}_1 \in D(T^{\lambda}), \; \nonumber\\
& \hat{\xi}_2(\k)=  \f{c}{|\k|^2 +1 } \Big( \beta \sin \big(s_0 \log |\k| \big) + \cos \big( s_0 \log |\k|\big) \Big)\Big\}
\end{align}

\n
where $c$ is an arbitrary constant and $s_0$ is the positive solution of the equation

\be
1-\f{8}{\sqrt{3}}\, \, \f{ \sinh \f{\pi s}{6} }{s \,\cosh \f{\pi s}{2}}=0\,.
\ee

\n
One can observe that both the diagonal and the non diagonal parts of $T^{\lambda} $ diverge on functions with the asymptotic behaviour  of $\hat{\xi}_2$ in \eqref{dtb} for $|\k|\rightarrow \infty$, but their sum remains finite.

\n
Given the s.a. operator $T^{\lambda}_{\beta}$, $D(T^{\lambda}_{\beta})$,  one obtains the s.a. extension $H_{\alpha, \beta}$ (also called Ter-Martirosian, Skornyakov Hamiltonian) of $H_{\alpha}$
\begin{align}
&D(H_{\alpha,\beta}) = \Big\{\psi \in L_s^2(\R^6) \;|\; \psi=w^{\lambda} + \mg \xi, \; w^{\lambda} \in H^2(\R^6), \; \hat{\xi} \in D(T^{\lambda}_{\beta}), \label{dhab}\nonumber\\
&\;\;\;\;\;\;\;\;\;\;\; \;\;\;\;\;\;\;\;\; \alpha \,\hat{\xi}(\p)+\big( T^{\lambda}  \hat{\xi} \,\big) (\p) \;=\;  (w^{\lambda}\big|_{\x =0})^{\wedge} (\p) \Big\}, \\
&(H_{\alpha,\beta} + \lambda) \psi = (H_0 + \lambda ) w^{\lambda}\,,\label{ahab} 
\end{align}
where
\be
H_0 =- \Delta_{\x} - \f{3}{4}\Delta_{\y}\,, \;\;\;\;\;\;\;\; D(H_0)=H^2(\R^6)\,.
\ee

\n
Roughly speaking, $\beta$ parametrizes a further boundary condition satisfied at the triple coincidence point $\x_1=\x_2=\x_3=0$. 
Therefore, it can be considered as  the strength of a sort of an additional three-body force acting on the particles when all their positions coincide.

\n
The authors conclude claiming that some further results on the spectrum of the Hamiltonian $H_{\alpha, \beta}$  hold. In particular, they affirm that  $H_{\alpha, \beta}$ has the unphysical instability  property already noted by Danilov, i.e., that there exists an infinite sequence  of negative eigenvalues
\[
E_n \rightarrow -\infty \,, \;\;\;\; \text{as } \;\; n \rightarrow \infty\,.
\]
The rigorous proof of this fact is contained  in  their second paper on the subject  and it will be described below.

\n
At the end of the paper  one finds an interesting remark on the possibility to define a modified Hamiltonian satisfying the stability property, i.e., bounded from below. The authors say: 
\say{ We note that this last result (i.e., the instability property) somewhat discredits our chosen extension, since probably only semibounded energy operators are of interest in nonrelativistic quantum mechanics. 
It seems to us that there must exist among the other extensions of the operator $\tilde{H}_0$ semibounded extensions which have all the properties of the model of Ter-Martirosian and    Skornyakov  that are good from the physical point of view, namely the properties of locality and of the correct character of the continuous spectrum. 
Evidently such extensions will be obtained ... } if one replaces the constant $\alpha$ in \eqref{bc22} (or equivalently in \eqref{bc23}) with the operator  $\alpha_M$ in the Fourier space defined by
\be\label{alk}
(\alpha_M \hat{\xi})(\p)= \alpha \hat{\xi} (\p)+ (K \hat{\xi})(\p)
\ee

\n
where $\alpha \in \R$ and $K$ is a convolution operator with a kernel $K(\p-\p')$  having the asymptotic behavior

\be\label{defk}
K(\p) \sim \f{\gamma}{\,|\p|^2} \,, \;\;\;\;\; \text{for} \;\;\;\; |\p| \rightarrow \infty 
\ee

\n
with the constant $\gamma$ satisfying 
\be\label{bgamma}
\gamma > \f{1}{\pi^3} \Big( \f{4\pi}{3 \sqrt{3}} -1 \Big) \,.
\ee
Unfortunately, they conclude: 
\say{A detailed development of this point of view is not presented here because of lack of space.}

\n
We believe that  such suggestion is interesting and we find it rather strange that the idea has never been developed in the literature. We also  observe that it is not so evident that the replacement of the constant $\alpha$ with the operator $\alpha_M$ defined in \eqref{alk}, \eqref{defk}, \eqref{bgamma} produces a semibounded Hamiltonian and, moreover, it is not clear the physical meaning of such replacement. We shall come back to this point in the next section. 

\n
Here we continue the analysis of the contribution of Minlos and Faddeev discussing the content of their second paper on the subject (\cite{mf2}), where the authors show that  the Hamiltonian $H_{\alpha, \beta}$ has an infinite number of eigenvalues accumulating both at zero and at $-\infty$.   We give here a  slightly different, and more elementary, proof than the one given in \cite{mf2}. To simplify the notation,  we consider only the case $\alpha=0$. We recall that $\alpha=0$ corresponds to a two-body interaction with zero-energy resonance (see, e.g., \cite{al}).

\n
Taking into account  definitions \eqref{dhab}, \eqref{ahab}, an eigenvector of $H_{0,\beta}$ associated to the negative eigenvalue $E=-\mu$, $\mu>0$, has the form $\mgm \xi$, where $\hat{\xi} \in D(T_{\beta})$ is a solution of the equation

\be\label{eqei1}
 \sqrt{\f{3}{4} |\p|^2 \!+\! \mu}  \,\, \hat{\xi}(\p) - \f{1}{\pi^2} \!\!\int\!\! d\p'\, \f{ \hat{\xi}(\p')}{|\p|^2 + |\p'|^2 + \p\cdot \p' +\mu} =0\,.
\ee

\n
We shall compute the rotationally invariant solutions  $\hat{\xi}= \hat{\xi}(|\p|)$ of \eqref{eqei1}.  Performing the angular integration in the second term of \eqref{eqei1}, one obtains the equation

\be\label{eqxi}
 \sqrt{\f{3}{4} p^2 \!+\! \mu}  \,\,p\, \hat{\xi}(p) - \f{2}{\pi} \int_0^{\infty} \!\!dp'\, p'\hat{\xi} (p') \, \log \f{p^2+p'^2 + pp'+\mu}{p^2 + p'^2 - pp'+\mu}  =0 \,.
\ee

\vs

\n
Let us introduce a change of the independent variable 
\be\label{chv}
p= \f{2\sqrt{\mu}}{\sqrt{3}} \, \sinh x\,, \;\;\;\;\;\;\;\;\;\; x= \log \left( \f{\sqrt{3} p}{2 \sqrt{\mu}} + \sqrt{\f{3p^2}{4\mu} +1} \right)
\ee
and define
\be\label{thx}
\theta(x) = 
\begin{cases} \mu \sinh x \cosh x \, \hat{\xi}\! \left( \f{2\sqrt{\mu}}{\sqrt{3}} \sinh x \!\right) \;\;\; &\text{for}\;\; x\geq0\;\;\;   \\
-\theta(-x) \;\;\; &\text{for}\;\;\; x<0 \,
\end{cases}
\ee
so that
\be\label{xith}
\hat{\xi}(p)= \f{2}{\sqrt{3}} \, \f{ \theta \! \left[ \log \left( \f{\sqrt{3} p}{2 \sqrt{\mu}} + \sqrt{\f{3p^2}{4\mu} +1} \right) \right]}{ p \,\sqrt{\f{3}{4} p^4 + \mu}} \,.
\ee

\n
The first term in \eqref{eqxi} in the new coordinates is 
\begin{align}\label{dixi}
\sqrt{\f{3}{4} p^2 \!+\! \mu}  \,\,p\, \hat{\xi}(p)
&= 
\sqrt{\mu\sinh^2 \!x + \mu} \;\, \f{2 \sqrt{\mu}}{\sqrt{3}} \sinh x \; \hat{\xi}\!\left( \f{2 \sqrt{\mu}}{\sqrt{3}} \sinh x \!\right) \!\nonumber\\
&= \f{2}{\sqrt{3}} \, \theta(x)\,.
\end{align}

\n
The other term in \eqref{eqxi} in the new coordinates reads
\begin{align}\label{ndxi}
&- \f{2}{\pi} \int_0^{\infty} \!\!dp'\, p'\hat{\xi} (p') \, \log \f{p^2+p'^2 + pp'+\mu}{p^2 + p'^2 - pp'+\mu}\nonumber\\
&= - \f{8}{3\pi} \int_0^{\infty} \!\!dy\, \theta(y)  \log \f{\sinh^2 x + \sinh^2 y + \sinh x \sinh y + 3/4}{\sinh^2 x + \sinh^2 y - \sinh x \sinh y + 3/4}\nonumber\\
&= - \f{8}{3\pi} \int_0^{\infty} \!\!dy\, \theta(y)  \log \f{ \left( 2 \cosh (x+y) -1\right) \left(2 \cosh (x-y) + 1\right)}{ \left( 2 \cosh (x+y) +1\right) \left(2 \cosh (x-y) - 1\right)}\nonumber\\
&= - \f{8}{3\pi} \int_{-\infty}^{+\infty} \!\!dy\, \theta(y)  \log \f{  2 \cosh (x-y) + 1}{  2 \cosh (x-y) - 1}
\end{align}
where in the last line we have used the extension $\theta(x)=-\theta(-x)$ for $x<0$. Therefore, equation \eqref{eqxi} for $\hat{\xi}(p)$ is transformed into  the following convolution equation for $\theta(x)$

\be
\theta(x) -\f{4}{\sqrt{3} \pi} \int_{-\infty}^{+\infty} \!\!\! dy \, \theta(y) \, \log \f{2 \cosh (x-y) +1}{2 \cosh (x-y)-1} =0\,.
\ee

\n
Finally, we compute the Fourier transform (see \cite{er}, p. 36) and we arrive at the equation for $\hat{\theta}$

\be\label{eqht}
\left( 1 - \f{8}{\sqrt{3}} \, \f{\sinh \f{\pi}{6} s }{ s\, \cosh \f{\pi}{2} s} \right) \! \hat{\theta}(s)=0 \,.
\ee

\n
Denote by $g(s)$ the function in parenthesis  in \eqref{eqht}. It is easy to see that $g$ is even, monotone increasing for $s>0$ and $g(s) \rightarrow 1$ for $s \rightarrow +\infty$. Moreover, $g(0)= 1- \f{4\pi}{3 \sqrt{3}} <0$ and we conclude that the equation $g(s)=0$ has two solutions $s=\pm s_0$, with $s_0>0$.   Since $\hat{\theta}$ is an odd function, the solution of \eqref{eqht} reads

\be
\hat{\theta} (s) = \delta(s-s_0) - \delta(s+s_0)
\ee

\n
apart from a multiplicative constant and therefore

\be
\theta(x) = \sin s_0 x\,.
\ee

\n
From \eqref{xith} we obtain the solution of equation \eqref{eqxi} for any $\mu >0$

\be\label{sotms}
\hat{\xi}_{\mu}(p) = \f{\sin \!\left[ s_0 \log \left( \f{\sqrt{3}p}{2\sqrt{\mu}} + \f{1}{\sqrt{\mu}}  \sqrt{\f{3}{4} p^2 + \mu} \, \right) \right] }{p \,\sqrt{ \f{3}{4} p^2 + \mu  }  }\, .
\ee

\n
We note that the solution \eqref{sotms} belongs to $L^2(\R^3)$ but it does not belong to $D(T^{\lambda})$ because of the  behavior $O(p^{-2})$ for $p \rightarrow \infty$. On the other hand we can find suitable values of $\mu$ such that the solution  belongs to $ D(T^{\lambda}_{\beta})$.

\n
Indeed, denoting $\ve = \f{4\mu}{3}\,p^{-2}$, we have

\begin{align}\label{asbe}
\hat{\xi}_{\mu}(p) &= \f{\sin \left\{ s_0 \left[ \log  \f{\sqrt{3}p}{\sqrt{\mu}} + \log \left( 1 + \f{\ve/2}{1 + \sqrt{1 + \ve}} \right)  \right]\right\} }{p^2 \,\sqrt{  1+ \ve  }  } \nonumber\\
&= \f{ \sin \left( s_0 \log p + \f{s_0}{2} \log \f{3}{\mu} \right)}{p^2 \sqrt{1 + \ve} } +\eta_1(p)  
\nonumber\\
& = \cos \left( \f{s_0}{2} \log \f{3}{\mu} \right) \f{\sin \left(s_0 \log p\right)}{p^2 +1} +
  \sin \left( \f{s_0}{2} \log \f{3}{\mu} \right) \f{\cos \left(s_0 \log p\right)}{p^2+ 1}+\eta_2(p)  \,
\end{align}
where $\eta_1, \eta_2 \in D(T^{\lambda})$. 
According to \eqref{dtb}, in order to have $\hat{\xi}_{\mu} \in D(T^{\lambda}_{\beta})$  we  impose the condition

\be\label{conort}
 \cos \left( \f{s_0}{2} \log \f{3}{\mu} \right) = \beta \sin \left( \f{s_0}{2} \log \f{3}{\mu} \right).
\ee

\n
Condition \eqref{conort} is satisfied if and only if $\mu$ is equal to

\be\label{mun}
\mu_n= 3 \, e^{- \f{2}{s_0} \cot^{-1} \beta} \, e^{ \f{2\pi}{s_0}  n } , \;\;\;\;\;\; n\in \mathbb{Z}\,.
\ee

\n
Thus we  obtain an infinite sequence of negative eigenvalues 
\be
E_n= - \mu_n \,, \;\;\;\;\;\; n\in \mathbb{Z}
\ee
with corresponding eigenvectors $\mathcal G^{\mu_n} \xi_{\mu_n}$, where $\hat{\xi}_{\mu_n}$ is given by \eqref{sotms}.

\n
We stress that  the model Hamiltonian $H_{0,\beta}$ exhibits the Efimov effect, i.e., there exists an infinite sequence of eigenvalues $E_n \rightarrow 0$ for $n\rightarrow -\infty$ satisfying  the (exact) geometrical law 
\be\label{geol}
\f{E_n}{E_{n+1}} = \, e^{- \f{2\pi}{s_0}}\,. 
\ee
On the other hand, one also has    $E_n \rightarrow - \infty$ for $n\rightarrow +\infty$, corresponding to the instability property known as  Thomas effect. 


\section{{\bf Regularized zero-range interactions}}

In this section we propose a model Hamiltonian for three bosons with zero-range interactions regularized around the triple coincidence point $\x_1=\x_2=\x_3=0$ in such a way to avoid the Thomas effect, i.e., with a spectrum bounded from below. We follow the proposal contained in \cite{al3}  which, in turn, essentially coincides with the already mentioned suggestion discussed at the end of  \cite{mf1}.

\n
In the first part of \cite{al3} the authors announce an interesting mathematical result on the Efimov effect. They consider  a three-particle system with Hamiltonian $H$, 
with two-body, spherically symmetric, short range potentials such that  at least two of the two-body subsystems have zero-energy resonances (i.e. infinite scattering length). They claim that  $H$  has infinitely many spherically symmetric bound states with energy $E_n \rightarrow 0$ such that
\be\label{Efi}
\lim_{n\rightarrow \infty} \f{E_n}{E_{n+1}} = e^{-\f{2 \pi}{ \sigma}}
\ee
where $\sigma>0$ depends only on the mass ratios (and coincides with $s_0$ in \eqref{geol} in the case of three identical bosons). The  result should follow \say{from a detailed study of the asymptotic behavior of the action of the scaling group in the spaces of two and three-body Hamiltonians restricted to functions invariant under the natural action of $SO(3)$}, where the scaling  group is given by 
\be
H \rightarrow H_{\ve} = \f{1}{\ve^2} U_{\ve} H U^{-1}_{\ve} \,, \;\;\;\;\;\; (U_{\ve} \psi)(\x)= \f{1}{\ve^{3/2}} \psi \Big(\f{\x}{\ve} \Big)\,, \;\;\;\;\;\; \ve >0\,. 
\ee
Roughly speaking, for $\ve \rightarrow 0$  
the rescaled Hamiltonian should converge to the zero-range model $H_{0, \beta}$. Since  
for $H_{0, \beta}$ the Efimov effect, together with the property  $\f{E_n}{E_{n+1}} = e^{-2 \pi/\sigma}$, is explicitly verified, one should infer 
 the result for $H$. 
Unfortunately, this program has not been realized and it remains as a challenging open problem.

\n
In the last part of the paper the authors add an interesting remark  for our aim  concerning the  construction of a reasonable three-body Hamiltonian with zero-range interactions which is bounded from below.

\n
Indeed, they claim that one can consider a Hamiltonian with zero-range interactions where the  zero-range force between two particles depends on the position of the third one. If such three-body force is suitably chosen then the Hamiltonian is bounded from below.

\n
The  proposed recipe can be rephrased in the following way: in the boundary condition \eqref{bc23} one replaces the  constant $\alpha$  with a position-dependent term 
\be\label{alfaa}
\alpha_A(\y) =\alpha+  \f{\delta}{|\y|}\,, \;\;\;\;\;\;\;\;\; \alpha, \delta \in \R\,.
\ee
In the case of equal masses, they affirm that for 
\be\label{boundde}
\delta > \f{2}{\pi^2} \Big(  \f{4\pi}{3 \sqrt{3}} -1  \Big) 
\ee
the corresponding zero-range Hamiltonian  is bounded from below. Also the proof of this statement is postponed to a forthcoming paper which has never been published.

\n
Let us briefly comment on the above proposal. The  replacement of the constant $\alpha$ with a function $\alpha(\y)$, with $\alpha(\y) \rightarrow \infty$ for $|\y| \rightarrow 0$, has a  reasonable physical meaning. It means that when the positions of the three particles coincide, i.e., for $\x=\y=0$, the two-body interactions are switched off ($\alpha=\infty$ means no interaction). In this way one compensates the tendency of the three interacting particles to \say{fall in the center}. On the other hand, the specific choice of the function  $\alpha_A(\y) $ is not explained in the paper but one can imagine that  such a function allows some explicit computations (as in the case of the choice \eqref{defk} of the operator $K$ in  \cite {mf1}). 

\n
It is also natural to compare the two proposals contained in \cite{mf1} and  \cite{al3}. 
It is immediate to realize  that the two proposals essentially coincide in the sense that  one is the Fourier transform of the other, i.e., 
\be\label{relma}
(\alpha_A \xi)^{\!\wedge} (\p)= (\alpha_M \hat{\xi})(\p)\,, \,\;\;\;\;\; \text{if} \;\;\; \delta =  2 \pi^2 \gamma\,.
\ee
It is also important to stress that only  the asymptotic behavior of $K(\p)$ for $|\p| \rightarrow \infty$ (see \eqref{defk}) is relevant to obtain a lower bounded Hamiltonian, as it is correctly pointed out in \cite{mf1}. Correspondingly, it must be sufficient to require  only  the asymptotic behavior  $\delta |\y|^{-1} + O(1)$ for $|\y| \rightarrow 0$ for the position-dependent strength of the interaction.

\n
In conclusion, following the (common) idea proposed in \cite{mf1} and \cite{al3}, we introduce the Hamiltonian $H_{\tilde{\alpha}}$ characterized by the  boundary condition
\beq\label{bc23y}
\psi(\x,\y) = \f{\xi(\y)}{|\x|} + \tilde{\alpha}(|\y|) \, \xi(\y) + o(1) \,, \;\;\;\; \text{for}\;\; |\x|\rightarrow 0\, \;\; \text{and}\;\;\y \neq 0
\eeq
where
\beq\label{alfafin}
\tilde{\alpha} \;:\; \R_+ \rightarrow \R\,, \;\;\;\;\;\;\tilde{ \alpha}(r)= \alpha + \f{\delta}{r} \chi^\ell(r)
\eeq
with
\beq
 \alpha \in \R, \;\;\;\; \delta >0, \;\;\;\;\ell >0,   \quad \chi^\ell(r) =  \left\{ \begin{aligned} &1 \qquad & r \leq \ell \\
							 &  0 & r > \ell &\,.
			         \end{aligned}\right.  
\eeq 

\n
More precisely, we define the Hamiltonian  as follows
\begin{align}
&D(H_{\tilde{\alpha}}) = \Big\{\psi \in L_s^2(\R^6) \;|\; \psi=w^{\lambda} + \mg \xi, \; w^{\lambda} \in H^2(\R^6), \; \xi \in H^1(\R^3), \label{dh}\nonumber\\
&\;\;\;\;\;\;\;\;\;\;\; \;\;\;\; \;\;\; (\tilde{\alpha} \, \xi )^{\!\wedge}  (\p)+\big( T^{\lambda}  \hat{\xi} \,\big) (\p) \;=\; \big(w^{\lambda}\big|_{\x =0}\big)^{\!\wedge} (\p) \Big\}, \\
&(H_{\tilde{\alpha}} + \lambda) \psi = (H_0 + \lambda ) w^{\lambda}\,. \label{ah} 
\end{align}

\n
In a paper in preparation (\cite{bcft}) it will be proved  the following result.

\begin{pro}
Let $\delta > \delta_0$, where
\be
\delta_0 = \f{\sqrt{3}}{\pi} \left( \f{4\pi}{3\sqrt{3}}-1 \right).
\ee
 Then for any  $\ell \in (0,+\infty]$ and $\alpha \in \R$  the operator \eqref{dh}, \eqref{ah} is s.a. and bounded from below. 

\end{pro}

\n
The result shows that it is sufficient to add a three-body force with an arbitrary small (but different from zero) range to avoid the collapse. We stress that it would be interesting  to prove that  boundedness from below is preserved taking  \say{some suitable limit $\ell \rightarrow 0$}. 

\n
The proof is based on the analysis of the quadratic form associated to $H_{\tilde{\alpha}}$. 
Taking into account of  \eqref{dh}, \eqref{ah},  by an explicit computation for $\psi \in D(H_{\tilde{\alpha}})$ one obtains
\ba
(\psi, (H_{\tilde{\alpha}} + \lambda) \psi)&=&(\psi, (H_0 + \lambda) w^{\lambda})  \!=\! (w^{\lambda}, (H_0 + \lambda) w^{\lambda}) \! +\! (\mathcal G^{\lambda} \xi, (H_0 + \lambda) w^{\lambda}) \nonumber\\
&=& (w^{\lambda}, (H_0 + \lambda) w^{\lambda})  + 12 \pi \left[ (\xi, \tilde{\alpha}\, \xi) + (\hat{\xi}, T^{\lambda} \hat{\xi}) \right]\,. 
\ea

\n
Hence we define the quadratic form
\ba
F_{\tilde{\alpha}}(\psi) =(w^{\lambda}, (H_0 + \lambda) w^{\lambda})   - \lambda \|\psi\|^2  + 12 \, \pi \,  \Phi^{\lambda}_{\tilde{\alpha}} (\xi)
\ea
where
\be
\Phi^{\lambda}_{\tilde{\alpha}} (\xi)= \int\!\!d\y\, \tilde{\alpha} (|\y|) |\xi(\y)|^2 + \int\!\! d\p\, \overline{\hat{\xi} (\p)} (T^{\lambda} \hat{\xi} ) (\p) 
\ee
and
\be
D(F_{\tilde{\alpha}})= \Big\{    \psi \in L_s^2(\R^6) \;|\; \psi=w^{\lambda} + \mg \xi, \; w^{\lambda} \in H^1(\R^6), \; \hat{\xi} \in  H^{1/2} (\R^3) \Big\} \,.
\ee
The proof proceeds taking  such a quadratic form as starting point and proving  that it is closed and bounded from below. Therefore it uniquely defines a s.a. and bounded from below operator which coincides with \eqref{dh}, \eqref{ah}.


\section{On the  negative eigenvalues}

In this section we consider the eigenvalue problem for  $H_{\tilde{\alpha}}$ in the special case
\be
\alpha=0, \;\;\;\;\;\; \ell =+ \infty \,.
\ee
As we already noticed in the case of the TMS Hamiltonian, an eigenvector  associated to the negative eigenvalue $E=-\mu$, $\mu>0$, has the form $\mgm \xi$, where $\xi \in H^1(\R^3)$ is a solution of the equation

\be\label{eqei1?}
\f{\delta}{2 \pi^2} \int\!\! d\p'\, \f{\hat{\xi}(\p')}{|\p - \p'|^2}  +  \sqrt{\f{3}{4} |\p|^2 \!+\! \mu}  \,\, \hat{\xi}(\p) - \f{1}{\pi^2} \!\!\int\!\! d\p'\, \f{ \hat{\xi}(\p')}{|\p|^2 + |\p'|^2 + \p\cdot \p' +\mu} =0\,.
\ee
where the first term in \eqref{eqei1?} is the Fourier transform of $\delta |\y|^{-1} \xi(\y)$. 

\n
Proceeding as in section 2, we consider  the rotationally invariant case  $\hat{\xi}= \hat{\xi}(|\p|)$.  
Performing the angular integration one obtains the equation

\begin{align}\label{eqxi?}
 \!\!\f{\delta}{\pi} \!\int_0^{\infty} \!\!dp'\, p'\hat{\xi}(p')  \log \f{p+p'}{|p-p'|} &+ \sqrt{\f{3}{4} p^2 \!+\! \mu}  \,\,p\, \hat{\xi}(p) \nonumber\\
& - \f{2}{\pi} \int_0^{\infty} \!\!dp'\, p'\hat{\xi} (p') \, \log \f{p^2+p'^2 + pp'+\mu}{p^2 + p'^2 - pp'+\mu}  =0 \,.
\end{align}

\n
In the following we prove that for $\delta>\delta_0$  there are no  solutions of \eqref{eqxi?} and therefore the Hamiltonian has no negative eigenvalues corresponding to rotationally invariant solutions of \eqref{eqei1?}. The main point of the proof is that the l.h.s. of \eqref{eqxi?} can be diagonalized by the same change of coordinates   used in section 2 for the TMS Hamiltonian.

\begin{pro}
Let $\delta > \delta_0$. Then equation \eqref{eqxi?} has only the trivial solution. 
\end{pro}

\begin{proof}
Using \eqref{chv}, \eqref{thx}, for the first term in \eqref{eqxi?} we have

\begin{align}\label{dexi}
&\f{\delta}{\pi} \!\int_0^{\infty} \!\!dp'\, p'\hat{\xi}(p')  \log \f{p+p'}{|p-p'|} =
\f{4 \delta}{3 \pi} \int_0^{\infty} \!\!dy\, \theta(y) \log \left| \f{\sinh x + \sinh y}{\sinh x - \sinh y}\right| \nonumber\\
&= \f{4 \delta}{3 \pi} \int_0^{\infty} \!\!dy\, \theta(y) \log \left| \f{\sinh \f{x+y}{2} \cosh \f{x-y}{2}}{\cosh \f{x+y}{2} \sinh \f{x-y}{2}}\right| \nonumber\\
&= \f{4 \delta}{3 \pi} \int_0^{\infty} \!\!dy\, \theta(y) \log \left| \f{ \cosh \f{x-y}{2}}{ \sinh \f{x-y}{2}}\right| +
\f{4 \delta}{3 \pi} \int_0^{\infty} \!\!dy\, \theta(y) \log \left| \f{\sinh \f{x+y}{2} }{\cosh \f{x+y}{2} }\right| \nonumber\\
&=  \f{4 \delta}{3 \pi} \int_{-\infty}^{+\infty} \!\!dy\, \theta(y) \log \left| \coth \f{x-y}{2} \right|
\end{align}
where in the last line we have used the extension $\theta(x)=-\theta(-x)$ for $x<0$. Hence, by \eqref{dixi}, \eqref{ndxi} and \eqref{dexi}, in the new variables equation \eqref{eqxi?} reads
\begin{align}
 \theta(x)+ \f{2 \delta}{\sqrt{3} \pi} \int_{-\infty}^{+\infty} \!\!dy\, \theta(y) &\log \left| \coth \f{x-y}{2} \right| \nonumber\\ 
& -\f{4}{\sqrt{3} \pi} \int_{-\infty}^{+\infty} \!\!\! dy \, \theta(y) \, \log \f{2 \cosh (x-y) +1}{2 \cosh (x-y)-1} =0\,.
\end{align}
We note that 
\begin{align}
& \f{1}{\sqrt{2\pi}} \int \!\!dx\, e^{-isx} \log \left| \coth \f{x}{2} \right| =
\sqrt{\f{2}{\pi}} \int_0^{\infty} \!\! dx \, \cos sx \, \log \!\left( \coth \f{x}{2} \right) \nonumber\\
& = \sqrt{\f{2}{\pi}} \int_0^{\infty} \!\! dx \, \cos sx \, \log \!\left( 1 + e^{-x} \right) - \sqrt{\f{2}{\pi}} \int_0^{\infty} \!\! dx \, \cos sx \, \log \!\left( 1 - e^{-x} \right)\,.
\end{align}
Then we use \cite{GR}, page 582, and \eqref{eqht} to write the equation for the Fourier transform of $\theta$ 
\be\label{fus}
\left( 1 + 2 \,\, \f{\delta \sinh \f{\pi}{2} s - 4 \sinh \f{\pi}{6} s }{\sqrt{3} \, s  \cosh \f{\pi}{2} s } \right) \! \hat{\theta} (s) =0 \,.
\ee

\n
It remains to show that the function in parenthesis in \eqref{fus} is strictly positive for $\delta>\delta_0$.  
Note that  the function  is even and then we can consider $s\geq 0$. For $\delta >\delta_0$   it is positive in $s=0$. Moreover, if we denote 
\be
F(s)= \sqrt{3} \, s  \cosh \f{\pi}{2} s + 2\delta \sinh \f{\pi}{2} s - 8 \sinh \f{\pi}{6} s \,,
\ee
 for $\delta>\delta_0$ we have  
\begin{align}
F'(s) &= (\sqrt{3} + \pi \delta) \cosh \f{\pi}{2} s + \f{\sqrt{3} \,\pi}{2} \,s \sinh \f{\pi}{2} s - \f{4\pi}{3} \cosh \f{\pi}{6} s \nonumber\\
& > (\sqrt{3} + \pi \delta_0) \cosh \f{\pi}{2} s - \f{4\pi}{3} \cosh \f{\pi}{6} s \nonumber\\
&\geq \left(\sqrt{3} + \pi \delta_0 - \f{4\pi}{3} \right) \cosh \f{\pi}{6} s =0 \,
\end{align}
so that $F(s) >0$ for any $s$. It follows that equation \eqref{fus}, and then equation \eqref{eqxi?}, has only the trivial solution and this conclude the proof of the proposition. 
\end{proof}

\section{Acknowledgements by way of conclusion}
We have many reasons for being grateful to Sergio Albeverio, mainly for his lasting friendship and for all the affection and help he has been giving us for so many years. Among all these good reasons, one concerns the support he provided to our scientific activity.\\
As it is very well known, he was one of the main players in the scientific project aimed to explicitly construct models in relativistic quantum field theory describing interacting  bosons, a project that never reached a final result in four dimensional space-time.  The formal non-relativistic limits of the relativistic models under study in constructive quantum field theory (see \cite{dim} for the only rigorous result in the investigation of such limits) describe particles in three dimensions interacting via zero-range forces.  This was the reason why, together with many other scientific interests, he got involved in the theory of point interaction Hamiltonians.\\
In the early nineties 
he suggested to one of us (A.T., young postdoc in Bochum at that time)  to investigate further the spectral structure of the point interaction Hamiltonian for a three-boson system, following the suggestions given in \cite{al3}. With some delay, together with a group of other younger italian researchers, we figured out better those suggestions and started working in this direction.\\
In this contribution we outlined the history of the main achievements in this research field and we sketched our first results.


\end{document}